\documentclass[]{article}
\usepackage[english]{babel}
\usepackage[autostyle,italian=guillemets]{csquotes}
\usepackage[backend=biber,style=numeric]{biblatex}
\usepackage{amssymb}
\usepackage{mathtools}
\usepackage{amsthm}
\usepackage[colorlinks=true, allcolors=blue]{hyperref}
\usepackage[section]{algorithm}
\usepackage{algpseudocode}
\usepackage{listingsutf8}
\usepackage{color}
\usepackage{pgfplots}

\theoremstyle{definition}
\newtheorem{definition}{Definition}[section]

\theoremstyle{plain}

\newtheorem{theorem}{Theorem}[section]

\newtheorem{proposition}{Proposition}[section]

\newcommand{\N}{\mathbb{N}}
\newcommand{\Z}{\mathbb{Z}}
\newcommand{\R}{\mathbb{R}}
\newcommand{\Q}{\mathbb{Q}}

\definecolor{javagreen}{rgb}{0.25,0.5,0.35} 
\definecolor{javapurple}{rgb}{0.5,0,0.35} 
\definecolor{javadocblue}{rgb}{0.25,0.35,0.75} 

\title{A square root algorithm faster than Newton's method for multiprecision numbers, using floating-point arithmetic}
\author{Fabio Romano}

\begin{document}

\maketitle

\begin{abstract}
	In this paper, an optimized version of classical Bombelli's algorithm for computing integer square roots is presented. In particular, floating-point arithmetic is used to compute the initial guess of each digit of the root, following similar ideas to those used in \cite[4.3.1]{knuth:taocp2} for division. A program with an implementation of the algorithm in Java is also presented, and its running time is compared with that of the algorithm provided by the Java standard library, which uses the Newton's method. From tests, the algorithm presented here turns out to be much faster.
\end{abstract}

\section{Introduction: Newton vs Bombelli}

\blockquote[Henry S. Warren Jr.]{Ah, the elusive square root, \\ it should be a cinch to compute. \\ But the best we can do \\ Is use powers of 2 \\ And iterate the method of Newt!\\}

Newton's method for multiprecision numbers uses the following version of Babylonian formula (cf. \cite[11.1]{warren:hackers_delight}):
\[
y_{k+1} = \left\lfloor \frac{1}{2} \left( y_k +  \left\lfloor \frac{x}{y_k} \right\rfloor \right) \right\rfloor
\]

Since the formula involves an integer division of the radicand $x$ by the root's approximation $y_k$, its computation requires $O(n^c)$ running time, where $n$ is the precision (i.e., the number of digits) of $x$, and $c \ge \log_2 3$, depending on which algorithm is used for division (cf. \cite{burnikel:fast_division}).

Since Newton's method has quadratic convergence, which means that the correct digits roughly double with each iteration, and since the square root has about half the digits of the radicand, then the running time of Newton's method is $O(n^c \log n) = C_N \cdot n^c \log n$, for some $C_N > 0$.

On the other hand, Bombelli's algorithm, which in \cite[11.1]{warren:hackers_delight} is denoted as the "shift-and-subtract algorithm", has a running time of $O(n^2) = C_B \cdot n^2$, for some $C_B > 0$, where again $n$ is the precision of the radicand.

Hence, by asymptotic analysis, Bombelli's algorithm seems to be faster than Newton's method. However, currently the most used algorithm for computing the integer square roots of multiprecision numbers is Newton's method.

The reason behind this fact is that Bombelli's algorithm gives a condition that the digit has to satisfy, but does not specify how to find that digit, so indeed the naive approach is only able to guess one bit of a digit per single evaluation of the condition, causing it to fall back into binary search, and thus making $C_B$ too big. The reasons behind these statements will be clarified after formalizing the algorithm.

As a consequence, the constant $n_0 > 0$ such that the following holds $\forall n \ge n_0$:
\[
	C_B \cdot n^2 \le C_N \cdot n^c \log n
\]
is too large for practical applications.

So, the aim of this paper is precisely to greatly decrease the constant $C_B$ by guessing a whole digit per single evaluation of the condition instead of a single bit, through a constant number of machine operations, in order to make Bombelli's algorithm far more efficient than Newton's method in practice.

\section{Basic definitions and properties}

Of course, in order to present the algorithm, we first need to define what the square root is. Here, zero is supposed to be in $\N$.

\begin{definition}[Square root]
	Given a number $x \in \N$, the (integer) square root function, denoted by $\lfloor \sqrt x \rfloor$, returns the number $y \in \N$ such that:
	\[
		y^2 \le x < (y + 1)^2
	\]
\end{definition}

\begin{definition}[Remainder of the square root]
	Given a number $x \in \N$, the remainder of its square root is the difference $x - \lfloor \sqrt x \rfloor^2$.
\end{definition}

\begin{proposition}\label{prop:rem_upper_bound}
	For each $x \in \N$, it holds that $x - \lfloor \sqrt x \rfloor^2 \le 2\lfloor \sqrt x \rfloor$.
\end{proposition}

\begin{proof}\mbox{}
	
	Since, by definition, $x < (\lfloor \sqrt x \rfloor + 1)^2$:
	\\$\implies x < \lfloor \sqrt x \rfloor^2 + 2\lfloor \sqrt x \rfloor + 1$
	\\$\implies x - \lfloor \sqrt x \rfloor^2 < 2\lfloor \sqrt x \rfloor + 1$
	\\$\implies x - \lfloor \sqrt x \rfloor^2 \le 2\lfloor \sqrt x \rfloor$, since $\lfloor \sqrt x \rfloor$ is an integer.
\end{proof}

\section{Bombelli's algorithm}

\begin{algorithm}[H]
	\caption{Shift-and-subtract}
	\label{alg:sqrt}
	\begin{algorithmic}[1]
		\Require a non-negative $2n$-place integer $X = (x_{2n-1} \dots x_0)_b$, where $b \ge 2$ is the radix of the number.
		\Ensure $\Call{Sqrt}{X} = (Y, R)$, where $Y = (y_{n-1} \dots y_0)_b$ is the square root of $X$ and $R = (r_n \dots r_0)_b$ is its remainder.
		
		\Function{Sqrt}{$X$}
			\If{$X = 0$}
				\State \Return $(0, 0)$
			\EndIf
			\\ \Comment{First iteration}
			\State $y_{n-1} \gets \max\{\beta \in \N \mid \beta^2 \le (x_{2n-1} x_{2n-2})_b\}$
			\State $Y_{n-1} \gets y_{n-1}$
			\State $R_{n-1} \gets (x_{2n-1} x_{2n-2})_b - y_{n-1}^2$
			
			\\ \Comment{Next iterations}
			\For{$i \gets n - 2$ \textbf{downto} $0$}
				\State $y_i \gets \max\{\beta \in \N \mid 2b Y_{i+1} \beta + \beta^2 \le (R_{i+1} x_{2i+1} x_{2i})_b\}$
				\State $Y_i \gets (Y_{i+1} y_i)_b$
				\State $R_i \gets (R_{i+1} x_{2i+1} x_{2i})_b - (2b Y_{i+1} y_i + y_i^2)$
			\EndFor
			
			\State \Return $(Y_0, R_0)$
		\EndFunction
	\end{algorithmic}
\end{algorithm}

The precision of the remainder is $n + 1$ due to proposition \ref{prop:rem_upper_bound}.

\begin{theorem}[Soundness]
	Given an input to the algorithm \ref{alg:sqrt} that satisfies its requirements, the output will satisfy its assurance.
\end{theorem}

\begin{proof}
	To prove the soundness, the following invariants are used:
	\begin{equation}
		Y_i^2 \le (x_{2n-1} \dots x_{2i})_b < (Y_i + 1)^2
	\end{equation}
	\begin{equation}
		R_i = (x_{2n-1} \dots x_{2i})_b - Y_i^2
	\end{equation}
	
	\begin{flushleft}
		\textbf{Initialization:} after the execution of line 8, the invariants hold for $i = n - 1$.
	\end{flushleft}
	
	Let $S := \{\beta \in \N \mid \beta^2 \le (x_{2n-1} x_{2n-2})_b\}$. Since $S$ is finite and non-empty and $y_{n-1} = \max(S)$, then $y_{n-1}$ has a (unique) value, so:
	\\$\implies y_{n-1}^2 \le (x_{2n-1} x_{2n-2})_b < (y_{n-1} + 1)^2$
	\\$\implies Y_{n-1}^2 \le (x_{2n-1} x_{2(n-1)})_b < (Y_{n-1} + 1)^2$, since $Y_{n-1} = y_{n-1}$
	\\$\implies Y_{i}^2 \le (x_{2n-1} x_{2i})_b < (Y_{i} + 1)^2$, for $i = n - 1$
	\\Thus, the first invariant holds.
	
	Furthermore, we have:
	\[
	y_{n-1}^2
	\le (x_{2n-1} x_{2n-2})_b
	\le b(b - 1) + (b - 1)
	= b^2 - 1
	\]
	Hence $y_{n-1} < b$, so its value can be a digit.
	
	For the second invariant, since $R_{n-1} = (x_{2n-1} x_{2n-2})_b - y_{n-1}^2$:
	\\$\implies R_{n-1} = (x_{2n-1} x_{2(n-1)})_b - Y_{n-1}^2$
	\\$\implies R_{i} = (x_{2n-1} x_{2i})_b - Y_{i}^2$
	\\Thus, the second invariant holds.
	
	\begin{flushleft}
		\textbf{Maintenance:} after each iteration of the for loop at line 10, the invariants still hold.
	\end{flushleft}
	
	By inductive hypothesis, the invariants hold for $i + 1$:
	\begin{equation}\label{eq:inv_1_plus1}
		Y_{i+1}^2 \le (x_{2n-1} \dots x_{2(i+1)})_b < (Y_{i+1} + 1)^2
	\end{equation}
	\begin{equation}\label{eq:inv_2_plus1}
		R_{i+1} = (x_{2n-1} \dots x_{2(i+1)})_b - Y_{i+1}^2
	\end{equation}

	Consider the following chain of equivalences, for any $\beta \in \N$:
	\begin{equation}\label{eq:equiv_digit_def}
		\begin{split}
			&(b Y_{i+1} + \beta)^2
			\le (x_{2n-1} \dots x_{2i})_b \\
			&\iff (b Y_{i+1} + \beta)^2
			\le b^2(x_{2n-1} \dots x_{2i+2})_b + (x_{2i+1} x_{2i})_b \\
			&\iff (b Y_{i+1} + \beta)^2
			\le b^2(x_{2n-1} \dots x_{2(i+1)})_b + (x_{2i+1} x_{2i})_b \\
			&\iff (b Y_{i+1} + \beta)^2
			\le b^2(Y_{i+1}^2 + R_{i+1}) + (x_{2i+1} x_{2i})_b, \text{ by \eqref{eq:inv_2_plus1}} \\
			&\iff b^2 Y_{i+1}^2 + 2b Y_{i+1} \beta + \beta^2
			\le b^2 Y_{i+1}^2 + b^2 R_{i+1} + (x_{2i+1} x_{2i})_b \\
			&\iff 2b Y_{i+1} \beta + \beta^2
			\le b^2 R_{i+1} + (x_{2i+1} x_{2i})_b \\
			&\iff 2b Y_{i+1} \beta + \beta^2
			\le (R_{i+1} x_{2i+1} x_{2i})_b
		\end{split}
	\end{equation}
	
	So, we have:
	\begin{equation}\label{eq:inv_1_left}
		2b Y_{i+1} \beta + \beta^2 \le (R_{i+1} x_{2i+1} x_{2i})_b
		\iff (b Y_{i+1} + \beta)^2 \le (x_{2n-1} \dots x_{2i})_b
	\end{equation}
	
	which is equivalent to:
	\begin{equation}\label{eq:inv_1_right}
		(R_{i+1} x_{2i+1} x_{2i})_b < 2b Y_{i+1} \beta + \beta^2
		\iff (x_{2n-1} \dots x_{2i})_b < (b Y_{i+1} + \beta)^2
	\end{equation}
	
	Now, let $S := \{\beta \in \N \mid 2b Y_{i+1} \beta + \beta^2 \le (R_{i+1} x_{2i+1} x_{2i})_b\}$.
	
	By \eqref{eq:inv_1_plus1} and \eqref{eq:inv_2_plus1}, $R_{i+1} = (x_{2n-1} \dots x_{2(i+1)})_b - Y_{i+1}^2 \ge 0$, so, if $\beta = 0$, then $2b Y_{i+1} \beta + \beta^2 \le (R_{i+1} x_{2i+1} x_{2i})_b$ holds, thus $S$ is non-empty, and since $y_i = \max(S)$, then $y_i$ has a (unique) value, so:
	\\$\implies 2b Y_{i+1} y_i + y_i^2 \le (R_{i+1} x_{2i+1} x_{2i})_b < 2b Y_{i+1} (y_i + 1) + (y_i + 1)^2$, by definition of $S$
	\\$\implies (b Y_{i+1} + y_i)^2 \le (x_{2n-1} \dots x_{2i})_b < (b Y_{i+1} + y_i + 1)^2$, by \eqref{eq:inv_1_left} and \eqref{eq:inv_1_right}
	\\$\implies Y_i^2 \le (x_{2n-1} \dots x_{2i})_b < (Y_i + 1)^2$, since $Y_i = (Y_{i+1} y_i)_b$
	\\Thus, the first invariant holds.
	
	Furthermore, we have:
	\[
		\begin{split}
			(b Y_{i+1} + y_i)^2
			&\le (x_{2n-1} \dots x_{2i})_b \\
			&= b^2 (x_{2n-1} \dots x_{2i+2})_b + (x_{2i+1} x_{2i})_b \\
			&\le b^2 (x_{2n-1} \dots x_{2i+2})_b + b(b - 1) + (b - 1) \\
			&= b^2 (x_{2n-1} \dots x_{2(i+1)})_b + b^2 - 1 \\
			&\le b^2 [(Y_{i+1} + 1)^2 - 1] + b^2 - 1, \text{ by } \eqref{eq:inv_1_plus1} \\
			&= b^2 (Y_{i+1} + 1)^2 - 1 \\
			&< b^2 (Y_{i+1} + 1)^2 \\
			&= (b Y_{i+1} + b)^2
		\end{split}
	\]
	Hence $y_i < b$, so its value can be a digit.
	
	For the second invariant:
	\[
		\begin{split}
			(x_{2n-1} \dots x_{2i})_b - Y_i^2
			&= (x_{2n-1} \dots x_{2i})_b - (b Y_{i+1} + y_i)^2, \text{ since } Y_i = (Y_{i+1} y_i)_b \\
			&= b^2 (x_{2n-1} \dots x_{2i+2})_b + (x_{2i+1} x_{2i})_b - (b^2 Y_{i+1}^2 + 2b Y_{i+1} y_i + y_i^2) \\
			&= b^2 [(x_{2n-1} \dots x_{2(i+1)})_b - Y_{i+1}^2] + (x_{2i+1} x_{2i})_b - (2b Y_{i+1} y_i + y_i^2) \\
			&= b^2 R_{i+1} + (x_{2i+1} x_{2i})_b - (2b Y_{i+1} y_i + y_i^2), \text{ by \eqref{eq:inv_2_plus1}} \\
			&= (R_{i+1} x_{2i+1} x_{2i})_b - (2b Y_{i+1} y_i + y_i^2) \\
			&= R_i, \text{ by line 13} \\
		\end{split}
	\]
	Thus, the second invariant holds.
	
	\begin{flushleft}
		\textbf{Termination:} after the last iteration of the for loop at line 10, the invariants hold for $i = 0$.
	\end{flushleft}
	
	\[
		Y_0^2 \le (x_{2n-1} \dots x_0)_b < (Y_0 + 1)^2
	\]
	\[
		R_0 = (x_{2n-1} \dots x_0)_b - Y_0^2
	\]
	
	By the requirements of the algorithm, $X = (x_{2n-1} \dots x_0)_b$, so $Y_0$ and $R_0$ satisfy the definition of square root of $X$ and remainder respectively.
\end{proof}

\subsection{Time complexity}

As we can see, the problematic part for the running time of the algorithm lies in line 11: using binary search to find the value of the digit $y_i$, almost always it involves evaluating the inequality $2b Y_{i+1} \beta + \beta^2 \le (R_{i+1} x_{2i+1} x_{2i})_b$ as many times as the number of bits of the radix $b$. Moreover, the time for evaluating the inequality, which is linear in the number of digits of the partial root $Y_{i+1}$, grows as the quantity $2b Y_{i+1} \beta + \beta^2$ approaches $(R_{i+1} x_{2i+1} x_{2i})_b$.

So, even if the time complexity is $O(n^2)$ in the precision of $X$, the algorithm is still extremely inefficient; we would like to have a method for guessing digits in constant time, such that the only operation that requires linear time during a single iteration is computing the remainder.

\section{Guessing the digits}

\subsection{The basic formula}

The general idea for guessing the square root digits is actually quite simple. We start from the definition for computing a digit, given by the algorithm:
\[
	y_i := \max\{\beta \in \N \mid 2b Y_{i+1} \beta + \beta^2 \le (R_{i+1} x_{2i+1} x_{2i})_b\}
\]

By the chain of equivalences \eqref{eq:equiv_digit_def}, this definition is the same as:
\[
	y_i = \max\{\beta \in \N \mid (b Y_{i+1} + \beta)^2
	\le b^2 Y_{i+1}^2 + b^2 R_{i+1} + (x_{2i+1} x_{2i})_b\}
\]

Now, consider the inequality in the definition. If we turn it into an equation and then we solve it for $\beta$, its integer part will be exactly the value we were looking for $y_i$, that is:
\[
	\begin{split}
		&(b Y_{i+1} + \beta)^2
		= b^2 Y_{i+1}^2 + b^2 R_{i+1} + (x_{2i+1} x_{2i})_b \\
		&\iff (b Y_{i+1} + \beta)^2
		= (b Y_{i+1})^2 + (R_{i+1} x_{2i+1} x_{2i})_b \\
		&\iff b Y_{i+1} + \beta
		= \sqrt{(b Y_{i+1})^2 + (R_{i+1} x_{2i+1} x_{2i})_b}, \text{ since } b Y_{i+1} + \beta \ge 0 \\
		&\iff \beta
		= \sqrt{(b Y_{i+1})^2 + (R_{i+1} x_{2i+1} x_{2i})_b} - b Y_{i+1}
	\end{split}
\]

Thus, the formula we were looking for to guess $y_i$ is:
\begin{equation}\label{eq:guess_formula}
	y_i = \left\lfloor \sqrt{(b Y_{i+1})^2 + (R_{i+1} x_{2i+1} x_{2i})_b} - b Y_{i+1} \right\rfloor
\end{equation}

But this is not yet the formula that will be used in the final algorithm; the reason for this will be clarified in the next section.

\subsection{Evaluate the formula using floating-point arithmetic}

Obviously, multiprecision arithmetic can't be used to evaluate the just derived formula, since it requires to compute the square root of a multiprecision number, but we are already building an algorithm to compute square roots of multiprecision numbers!

To get around this vicious circle, we can use floating-point numbers, and consequently use a fast implementation of the square root in floating-point arithmetic, often already included in the computer hardware.

However, the formula as it appears in the form of \eqref{eq:guess_formula} could lead to large rounding errors, resulting from the subtraction of two numbers that have close magnitudes. Therefore, the formula must be rewritten through the following algebraic transformations:
\[
	\begin{split}
		y_i
		&= \left\lfloor \sqrt{(b Y_{i+1})^2 + (R_{i+1} x_{2i+1} x_{2i})_b} - b Y_{i+1} \right\rfloor \\
		&= \left\lfloor \left( \sqrt{(b Y_{i+1})^2 + (R_{i+1} x_{2i+1} x_{2i})_b} - b Y_{i+1} \right) \cdot \frac{\sqrt{(b Y_{i+1})^2 + (R_{i+1} x_{2i+1} x_{2i})_b} + b Y_{i+1}}{\sqrt{(b Y_{i+1})^2 + (R_{i+1} x_{2i+1} x_{2i})_b} + b Y_{i+1}} \right\rfloor \\
		&= \left\lfloor \frac{(b Y_{i+1})^2 + (R_{i+1} x_{2i+1} x_{2i})_b - (b Y_{i+1})^2}{\sqrt{(b Y_{i+1})^2 + (R_{i+1} x_{2i+1} x_{2i})_b} + b Y_{i+1}} \right\rfloor \\
		&= \left\lfloor \frac{(R_{i+1} x_{2i+1} x_{2i})_b}{\sqrt{(b Y_{i+1})^2 + (R_{i+1} x_{2i+1} x_{2i})_b} + b Y_{i+1}} \right\rfloor \\
	\end{split}
\]

Of course, the denominator is non-zero, since the leading digit of the partial root $Y_{i+1}$ is always non-zero.
Thus, this is the final form of the formula:
\begin{equation}\label{eq:guess_fromula_final}
	y_i = \left\lfloor \frac{(R_{i+1} x_{2i+1} x_{2i})_b}{\sqrt{(b Y_{i+1})^2 + (R_{i+1} x_{2i+1} x_{2i})_b} + b Y_{i+1}} \right\rfloor
\end{equation}

\section{Error analysis}

\subsection{Definitions for error analysis}

\begin{definition}[Floating-point numbers]
	Given $p \ge 2$ and $e_{min} < 0 < e_{max}$, the set $\mathcal F \subseteq \Q$ of non-negative floating-point numbers is defined as:
	\[
		\mathcal F = \{0\}\cup\{m \cdot 2^e \in \Q\}
	\]
	where:
	\begin{itemize}
		\item $e$ is an integer such that $e_{min} \le e \le e_{max}$;
		
		\item $1 \le m < 2$ has one bit before the point, and at most $p - 1$ bits after.
	\end{itemize}
	
	Cf. \cite[2.1]{muller:floating_point_arithmetic}.
\end{definition}

\begin{definition}[Round toward $-\infty$]
	$\text{RD}(x)$ is the largest floating-point number less than or equal to $x \in \R$.
	
	Cf. \cite[2.2.1]{muller:floating_point_arithmetic}.
\end{definition}

\begin{definition}[Round to nearest]
	$\text{RN}(x)$ is the floating-point number that is the closest
	to $x \in \R$. If $x$ falls exactly halfway
	between two consecutive floating-point numbers, it is rounded to the only one of these two whose integral significand is even.
	
	Cf. \cite[2.2.1]{muller:floating_point_arithmetic}.
\end{definition}

\begin{definition}[Next up and next down]
	$\text{nextUp}(x)$ is the least floating-point number that compares greater than $x \in \mathcal F$. $\text{nextDown}(x) := -\text{nextUp}(-x)$.
	
	Cf. \cite[5.3.1]{ieee754}
\end{definition}

\begin{definition}[Signed relative error]
	Given $x \in \R \setminus \{0\}$ and an approximation $\widetilde{x} \in \R$, its signed relative error is:
	\[
		\varepsilon_x := \frac{\widetilde{x} - x}{x}
	\]
	
	From which it derives:
	\[
		\widetilde{x} = x(1 + \varepsilon_x)
	\]
	
	The (unsigned) relative error is simply $|\varepsilon_x|$ (cf. \cite[1.2]{higham:numerical_algorithms}, \cite[2.2.3]{muller:floating_point_arithmetic}).
\end{definition}

\begin{definition}[Model of arithmetic]
	The operations $\oplus, \odot, \oslash \colon \mathcal F \times \mathcal F \to \mathcal F$ are defined as follows:
	
	\begin{itemize}
		\item $x \oplus y = \text{RN}(x + y)$;
		
		\item $x \odot y = \text{RN}(x \cdot y)$;
		
		\item $x \oslash y = \text{RN}\left( \dfrac{x}{y} \right)$.
	\end{itemize}
	
	Cf. \cite[6.1.1]{muller:floating_point_arithmetic}, \cite[2.2]{higham:numerical_algorithms}.
\end{definition}

\subsection{Estimations of relative errors}

\subsubsection{Unit roundoff}

Given $x \in [2^e; 2^{e + 1})$ with $e \in \Z$ and $e_{min} \le e \le e_{max}$, the upper bound of its relative error due to RN rounding mode is:
\[
	\frac{|x - \text{RN}(x)|}{|x|} \le \frac{\frac{1}{2} 2^{e + 1 - p}}{2^e} = 2^{-p} = \textbf{u}
\]

\textbf{u} is called \emph{unit roundoff} (cf. \cite[2.6.4]{muller:floating_point_arithmetic}).

While the upper bound of the relative error due to RD rounding mode is:
\[
	\frac{|x - \text{RD}(x)|}{|x|} \le \frac{2^{e + 1 - p}}{2^e} = 2^{1 - p} = 2\textbf{u}
\]

Now, if we consider $\widetilde{x} = x(1 + \varepsilon_x)$, let's estimate how the relative error becomes if we change $\widetilde{x}$ with $\text{nextUp}(\widetilde{x})$ (or $\text{nextDown}(\widetilde{x})$):
\[
	\begin{split}
		\frac{|x - \text{nextUp}(\widetilde{x})|}{|x|}
		&\le \frac{|(x - \widetilde{x}) + 2^{e + 1 - p}|}{|x|} \\
		&= \left| \frac{x - \widetilde{x}}{x} + \frac{2^{e + 1 - p}}{x} \right| \\
		&\le \left| -\varepsilon_x + \frac{2^{e + 1 - p}}{2^e} \right| \\
		&\le |\varepsilon_x| + 2\textbf{u}
	\end{split}
\]

\subsubsection{Upper bounds for relative errors of arithmetic operations}

Consider $x, y \in (0; 2^{e_{max} + 1})$ and their approximations $\widetilde{x} = x(1 + \varepsilon_x)$ and $\widetilde{y} = y(1 + \varepsilon_y)$, such that $\varepsilon_x, \varepsilon_y \ll 1$ and $\widetilde{x}, \widetilde{y} \in \mathcal F$, and assume that no overflow or underflow occurs for the operations below.

For the sum, the upper bound of the relative error is:
\[
	\begin{split}
		\frac{|(x + y) - (\widetilde{x} \oplus \widetilde{y})|}{|x + y|}
		&= \frac{|(x + y) - [x(1 + \varepsilon_x) + y(1 + \varepsilon_y)] (1 + \varepsilon_s)|}{|x + y|} \\
		&= \frac{|x \varepsilon_x + y \varepsilon_y + x \varepsilon_s + x \varepsilon_x \varepsilon_s + y \varepsilon_s + y \varepsilon_y \varepsilon_s|}{|x + y|} \\
		&\approx \frac{|x \varepsilon_x + y \varepsilon_y + \varepsilon_s (x + y)|}{|x + y|} \\
		&= \left| \frac{x}{x + y} \varepsilon_x + \frac{y}{x + y} \varepsilon_y + \varepsilon_s \right| \\
		&\le |\varepsilon_x| + |\varepsilon_y| + |\varepsilon_s|, \text{ since } x, y \le x + y \\
		&\le |\varepsilon_x| + |\varepsilon_y| + \textbf{u}
	\end{split}
\]

For the product, the upper bound is:
\[
\begin{split}
	\frac{|(x \cdot y) - (\widetilde{x} \odot \widetilde{y})|}{|x \cdot y|}
	&= \frac{|(x \cdot y) - [x(1 + \varepsilon_x) \cdot y(1 + \varepsilon_y)] (1 + \varepsilon_p)|}{|x \cdot y|} \\
	&= |1 - (1 + \varepsilon_x) (1 + \varepsilon_y) (1 + \varepsilon_p)| \\
	&= |\varepsilon_x + \varepsilon_y + \varepsilon_p + \varepsilon_x \varepsilon_y + \varepsilon_x \varepsilon_p + \varepsilon_y \varepsilon_p + \varepsilon_x \varepsilon_y \varepsilon_p| \\
	&\approx |\varepsilon_x + \varepsilon_y + \varepsilon_p| \\
	&\le |\varepsilon_x| + |\varepsilon_y| + |\varepsilon_p| \\
	&\le |\varepsilon_x| + |\varepsilon_y| + \textbf{u}
\end{split}
\]

For the quotient, the upper bound is:
\[
\begin{split}
	\frac{\left| \dfrac{x}{y} - [\widetilde{x} \oslash \widetilde{y}] \right|}{\left| \dfrac{x}{y} \right|}
	&= \frac{\left| \dfrac{x}{y} - \dfrac{x(1 + \varepsilon_x)}{y(1 + \varepsilon_y)} (1 + \varepsilon_q) \right|}{\left| \dfrac{x}{y} \right|} \\
	&= \left| 1 - \dfrac{(1 + \varepsilon_x)}{(1 + \varepsilon_y)} (1 + \varepsilon_q) \right| \\
	&= \left| 1 - \dfrac{(1 + \varepsilon_x) (1 - \varepsilon_y)}{(1 + \varepsilon_y) (1 - \varepsilon_y)} (1 + \varepsilon_q) \right| \\
	&= \left| 1 - \dfrac{(1 + \varepsilon_x) (1 - \varepsilon_y)}{1 - \varepsilon_y^2} (1 + \varepsilon_q) \right| \\
	&\approx \left| 1 - (1 + \varepsilon_x) (1 - \varepsilon_y) (1 + \varepsilon_q) \right| \\
	&\approx |\varepsilon_x - \varepsilon_y + \varepsilon_q| \\
	&\le |\varepsilon_x| + |\varepsilon_y| + |\varepsilon_q| \\
	&\le |\varepsilon_x| + |\varepsilon_y| + \textbf{u}
\end{split}
\]

For the square root, the upper bound is:
\[
\begin{split}
	\frac{\left| \sqrt x - \text{RN}\left( \sqrt{\widetilde{x}} \right) \right|}{|\sqrt x|}
	&= \frac{\left| \sqrt x - \sqrt{x(1 + \varepsilon_x)} (1 + \varepsilon_r) \right|}{|\sqrt x|} \\
	&= \left| 1 - \sqrt{1 + \varepsilon_x} (1 + \varepsilon_r) \right| \\
	&\approx \left| 1 - \sqrt{1 + \varepsilon_x + \left( \frac{\varepsilon_x}{2} \right)^2} (1 + \varepsilon_r) \right| \\
	&= \left| 1 - \left( 1 + \frac{\varepsilon_x}{2} \right) (1 + \varepsilon_r) \right| \\
	&= \left| \frac{\varepsilon_x}{2} + \varepsilon_r + \frac{\varepsilon_x \varepsilon_r}{2} \right| \\
	&\approx \left| \frac{\varepsilon_x}{2} + \varepsilon_r \right| \\
	&\le \frac{|\varepsilon_x|}{2} + |\varepsilon_r| \\
	&\le \frac{|\varepsilon_x|}{2} + \textbf{u}
\end{split}
\]

\subsection{Required precision of floating-point arithmetic}

Assume that the floating-point arithmetic used is accurate enough to represent the integer values between 0 and $b$ exactly, and such that the guess of a digit differs by at most one unit from its actual value.

If the guess is greater than the actual value of the digit, then the remainder is negative, and so it become easier to recognize the rounding error. For this reason, we will ensure that the fraction in \eqref{eq:guess_fromula_final} will always be rounded up.

For performance reasons, assume that the quantity $(R_{i+1} x_{2i+1} x_{2i})_b$ will be rounded down when converted to a floating point number, while the other quantities will be rounded to nearest. This last assumption is made because many systems only implement the "round to nearest" rounding mode.

To ensure the fraction is rounded up, the numerator must be maximized and the numerator minimized, so, in order to use the floating-point arithmetic, \eqref{eq:guess_fromula_final} must become:
\begin{equation}\label{eq:guess_formula_float}
	\widetilde{y_i} = \left\lfloor \text{nextUp}(num \oslash den) \right\rfloor
\end{equation}

where:
\begin{itemize}
	\item $num := \text{nextUp}(\text{RD}( (R_{i+1} x_{2i+1} x_{2i})_b))$
	
	\item $den := \text{nextDown}(sqrt \oplus yShifted)$
	
	\item $sqrt := \text{nextDown}\left( \text{RN} \left( \sqrt{rad} \right) \right)$
	
	\item $rad := \text{nextDown}(yShiftedSqr \oplus \text{RD}( (R_{i+1} x_{2i+1} x_{2i})_b))$;
	
	\item $yShiftedSqr := \text{nextDown}(yShifted \odot yShifted)$;
	
	\item $yShifted := b \odot \text{nextDown}(\text{RN}(Y_{i+1}))$.
\end{itemize}

$b$ is not rounded because it is exactly representable. If we also assume that $b$ is a power of 2, then the product $b \odot \text{nextDown}(\text{RN}(Y_{i+1}))$ is exact, and we can also avoid rounding $yShifted$.

Let's estimate the relative error given by the floating point expression in \eqref{eq:guess_formula_float}:
\[
	\begin{split}
		|\varepsilon_{y_i}|
		&\le (|\varepsilon_{num}| + |\varepsilon_{den}| + \textbf{u}) + 2\textbf{u} \\
		&\le 2\textbf{u} + (|\varepsilon_{sqrt}| + |\varepsilon_{yShifted}| + \textbf{u} + 2\textbf{u}) + \textbf{u} + 2\textbf{u} \\
		&= |\varepsilon_{sqrt}| + |\varepsilon_{yShifted}| + 8\textbf{u} \\
		&\le \left( \frac{|\varepsilon_{rad}|}{2} + \textbf{u} + 2\textbf{u} \right) + (|\varepsilon_b| + (\textbf{u} + 2\textbf{u}) + \textbf{u}) + 8\textbf{u} \\
		&= \frac{|\varepsilon_{rad}|}{2} + |\varepsilon_b| + 15\textbf{u} \\
		&\le \frac{1}{2} (|\varepsilon_{yShiftedSqr}| + 2\textbf{u} + \textbf{u} + 2\textbf{u}) + 15\textbf{u} \\
		&\le \frac{1}{2} ((2|\varepsilon_{yShifted}| + \textbf{u} + 2\textbf{u}) + 2\textbf{u} + \textbf{u} + 2\textbf{u}) + 15\textbf{u} \\
		&\le \frac{1}{2} (2(|\varepsilon_b| + (\textbf{u} + 2\textbf{u}) + \textbf{u}) + \textbf{u} + 2\textbf{u} + 2\textbf{u} + \textbf{u} + 2\textbf{u}) + 15\textbf{u} \\
		&= \frac{1}{2} \cdot 16\textbf{u} + 15\textbf{u} \\
		&= 23\textbf{u} \\
	\end{split}
\]

Now, to ensure that the guess of a digit differs by at most one unit from its actual value, we impose that the greatest value of the absolute error must be within one unit. The value of the absolute error is maximum when the fraction assumes the greatest possible value, i.e. $b$:
\[
	\begin{split}
		&23\textbf{u} \cdot b \le \delta, \text{ with } \delta \in (0; 1) \\
		&\iff 23 \cdot 2^{-p} \cdot b \le \delta \\
		&\iff p \ge \log_2 b + \log_2 \delta^{-1} + \log_2 23
	\end{split}
\]

In the underlying Java implementation, we take $p = 53$ and $b = 2^{32}$, thus the absolute error is within $\delta = 23 \cdot 2^{-53} \cdot 2^{32} \approx 1.1 \cdot 10^{-5}$. It can be thought as the probability that $y_i = \widetilde{y_i} - 1$ holds.

\textbf{Note:} although $\delta$ can be made as small as desired, it can never be zero (this would require infinite precision), so, until proven the contrary (by exhaustive search, for example), in general the possibility of an error in guessing the digit is unavoidable, hence it must be taken into account.

If the remainder computed using $\widetilde{y_i}$ is negative, then $y_i = \widetilde{y_i} - 1$, so the quantity to add to it is:
\[
	\begin{split}
		[2b Y_{i+1} (y_i + 1) + (y_i + 1)^2] - [2b Y_{i+1} y_i + y_i^2]
		&= 2b Y_{i+1} (y_i + 1) + (y_i + 1)^2 - 2b Y_{i+1} y_i - y_i^2 \\
		&= 2b Y_{i+1} [(y_i + 1) - y_i] + (y_i + 1)^2 - y_i^2 \\
		&= 2b Y_{i+1} + y_i^2 + 2y_i + 1 - y_i^2 \\
		&= 2b Y_{i+1} + 2y_i + 1 \\
	\end{split}
\]

\section{A Java implementation of the algorithm}

In this implementation, multiprecision integers (the class "BigInt") are represented as in the Java's standard library: by arrays of 32-bit signed integers, although they are treated as if they were unsigned. Basically, each element of the array represents a digit in $2^{32}$-radix. The digits are in big-endian order (from the most to the least significant, starting from index zero), and there are no leading zeros. The arrays are indexed by 32-bit signed integers, so the maximum array length is $2^{31} - 1$.

Floating-point numbers (the subclass "Floating") are unsigned and represented by two fields: the exponent and the significand. Numbers are normalized: if the significand is non-zero, it must be at least one and strictly less than 2; if it is zero, the exponent must have the minimum representable value.

Floating-point overflows are avoided by using 64-bit signed integers (the type "long") to represent the exponents of the floating-point numbers, while the significand is represented using the built-in double precision floating-point numbers of Java (the type "double"), which reflect the standard "binary64" format of IEEE-754 (cf. \cite{ieee754}). Floating-point underflows are taken into account only where they could happen.

\begin{lstlisting}[language = Java]
	package sqrtTest;
	
	import java.math.BigInteger;
	import java.util.Arrays;
	
	public class BigInt {
		
		private final int[] mag;
		
		/**
		* The length of this integer in bits, or -1 if not computed.
		* To get the value, call {@link #bitLength()} method.
		*/
		private int bitLength;
		
		private static final long LONG_MASK = 0xffffffffL;
		
		private static final BigInt ZERO = new BigInt(new int[0], true);
		
		private BigInt(int[] mag, boolean trusted) {
			this.mag = trusted ? mag : stripLeadingZeros(mag, 0);
			this.bitLength = -1;
		}
		
		BigInt(int[] mag) {
			this(mag, false);
		}
		
		public BigInt(String s) {
			this(computeMag(s), true);
		}
		
		private static int[] computeMag(String s) {
			String bin = new BigInteger(s).toString(2);
			if (bin.equals("0"))
			return new int[0];
			
			final int len = bin.length();
			int[] mag = new int[((len - 1) >> 5) + 1];
			int offset = len & 0x1f;
			if (offset == 0)
			offset = 32;
			
			mag[0] = Integer.parseUnsignedInt(bin.substring(0, offset), 2);
			for (int i = 1; i < mag.length; i++, offset += 32)
			mag[i] = Integer.parseUnsignedInt(bin.substring(offset, offset + 32), 2);
			
			return mag;
		}
		
		public BigInt[] sqrtAndRemainder() {
			if (mag.length == 0)
				return new BigInt[] { ZERO, ZERO };
			
			int[] r = new int[mag.length];
			int[] sqrt = new int[(mag.length + 1) >>> 1];
			int rBegin = 0, rEnd;
			long digit;
			final long b = 0x100000000L; // 2^32
			
			// first iteration
			if ((mag.length & 1) == 0) { // even length
				// ensure round up of digit
				digit = (long) Math.nextUp(
				Math.sqrt(Math.nextUp((mag[0] & LONG_MASK) * (double) b + (mag[1] & LONG_MASK))));
				
				if (digit == b) // avoid overflow due to too high digit's rounding
				digit--;
				
				final long digitSqr = digit * digit;
				long diff = (mag[1] & LONG_MASK) - (digitSqr & LONG_MASK);
				r[1] = (int) diff;
				diff = (mag[0] & LONG_MASK) - (digitSqr >>> 32) + (diff >> 32);
				r[0] = (int) diff;
				rEnd = 2;
				
				// Check for the unlikely case of too high digit's rounding
				if (diff >> 32 != 0) {
					digit--;
					// add (digit + 1)^2 - digit^2 == 2*digit + 1
					final long sum = (r[1] & LONG_MASK) + (digit << 1) + 1;
					r[1] = (int) sum;
					r[0] = (int) ((r[0] & LONG_MASK) + (sum >>> 32));
				}
			} else { // odd length
				// the integer sqrt is exact, no need to round up
				digit = (long) Math.sqrt(mag[0] & LONG_MASK);
				
				final long diff = (mag[0] & LONG_MASK) - digit * digit;
				r[0] = (int) diff;
				rEnd = 1;
			}
			sqrt[0] = (int) digit; // store digit
			
			Floating sqrtF = new Floating(digit);
			
			// next iterations
			for (int sqrtLen = 1; sqrtLen < sqrt.length; sqrtLen++) {
				// 1. Update remainder's begin
				for (; rBegin < rEnd && r[rBegin] == 0; rBegin++);
				
				// 2. Append the next block to the remainder
				r[rEnd] = mag[rEnd];
				r[rEnd + 1] = mag[rEnd + 1];
				rEnd += 2;
				
				// 3. Guess the next digit
				// compute remainder's floating values
				final Floating rFloor = Floating.valueOf(r, rBegin, rEnd), rCeil = rFloor.nextUp();
				
				sqrtF = sqrtF.scaleByPowerOf2(32);
				// ensure round up of digit
				digit = rCeil.div(
				(sqrtF.mul(sqrtF).nextDown().add(rFloor).nextDown()).sqrt().nextDown()
				.add(sqrtF).nextDown()).nextUp().toLong();
				
				if (digit != 0) {
					// 4. Compute the remainder
					if (digit == b) // avoid overflow due to too high digit's rounding
					digit--;
					
					digit = computeRemainder(digit, sqrt, sqrtLen, r, rBegin, rEnd);
					
					// 5. Add the digit to sqrt's floating value, if its bits might be representable
					if (sqrtLen <= 2) {
						if (sqrtLen == 2)
						sqrtF = sqrtF.nextUp(); // restore the correct rounding
						
						sqrtF = sqrtF.add(new Floating(digit)).nextDown(); // ensure round up of next digits
					}
					
					// 6. Store digit
					sqrt[sqrtLen] = (int) digit;
				}
			}
			
			return new BigInt[] { new BigInt(sqrt, true), new BigInt(stripLeadingZeros(r, rBegin), true) };
		}
		
		private static long computeRemainder(long digit, int[] sqrt, final int sqrtLen,
		int[] r, final int rBegin, final int rEnd) {
			// Subtract common parts
			int rPos = rEnd - 1;
			long prod = digit * digit;
			long diff = (r[rPos] & LONG_MASK) - (prod & LONG_MASK);
			r[rPos--] = (int) diff;
			
			final long twiceDigit = (digit << 1) & LONG_MASK;
			final boolean high = ((int) digit) < 0;
			
			// rEnd >= 2*sqrtLen + 1, no risk that rPos becomes negative
			for (int i = sqrtLen - 1; i >= 0; i--, rPos--) {
				prod = twiceDigit * (sqrt[i] & LONG_MASK) + (prod >>> 32);
				if (high)
				prod += sqrt[i + 1] & LONG_MASK;
				
				diff = (r[rPos] & LONG_MASK) - (prod & LONG_MASK) + (diff >> 32);
				r[rPos] = (int) diff;
			}
			
			// Borrow propagation
			if (rPos >= rBegin) {
				prod >>>= 32;
				if (high)
				prod += sqrt[0] & LONG_MASK;
				
				diff = (r[rPos] & LONG_MASK) - prod + (diff >> 32);
				r[rPos--] = (int) diff;
				
				if (rPos >= rBegin) {
					diff = (r[rPos] & LONG_MASK) + (diff >> 32);
					r[rPos] = (int) diff;
				}
				// The end: rEnd - rBegin <= sqrtLen + 3
			}
			
			// Check for the unlikely case of too high digit's rounding
			if (diff >> 32 != 0) {
				digit--;
				
				// Add common parts
				rPos = rEnd - 1;
				long sum = (r[rPos] & LONG_MASK) + (digit << 1) + 1;
				r[rPos--] = (int) sum;
				
				for (int i = sqrtLen - 1; i >= 0; i--, rPos--) {
					sum = (r[rPos] & LONG_MASK) + ((sqrt[i] & LONG_MASK) << 1) + (sum >>> 32);
					r[rPos] = (int) sum;
				}
				
				// Carry propagation
				if (rPos >= rBegin) {
					sum = (r[rPos] & LONG_MASK) + (sum >>> 32);
					r[rPos--] = (int) sum;
					
					if (rPos >= rBegin)
					r[rPos] = 0; // remainder's length can't be larger than sqrtLen + 2
					// The end: rEnd - rBegin <= sqrtLen + 3
				}
			}
			
			return digit;
		}
		
		private static class Floating {
			/**
			* <p> {@code this == signif * 2^exp}
			* <p> if {@code this != 0}, then {@code 1 <= signif < 2}
			*/
			final double signif;
			final long exp;
			
			static final Floating ZERO = new Floating(0.0, Long.MIN_VALUE);
			
			static final Floating MIN_VALUE = new Floating(1.0, ZERO.exp);
			
			private Floating(double signif, long exp) {
				this.signif = signif;
				this.exp = exp;
			}
			
			/**
			* Assume n is unsigned
			*/
			Floating(long n) {
				if (n == 0) {
					signif = ZERO.signif;
					exp = ZERO.exp;
				} else {
					exp = 63 - Long.numberOfLeadingZeros(n);
					final double val = n >= 0 ? n : n + 0x1p64;
					signif = exp == 0L ? val : val * Double.parseDouble("0x1p-" + exp);
				}
			}
			
			/**
			* For reasons of speed, if the argument cannot be represented exactly,
			* the magnitude is simply truncated: least significant bits are just discarded.
			*/
			static Floating valueOf(int[] mag, int begin, int end) {
				final int len = end - begin;
				if (len == 0)
				return ZERO;
				
				if (len <= 2) {
					long val = mag[begin] & LONG_MASK;
					if (len == 2)
					val = (val << 32) | (mag[begin + 1] & LONG_MASK);
					
					return new Floating(val);
				}
				
				final int bitLen = (len << 5) - Integer.numberOfLeadingZeros(mag[begin]);
				
				/*
				* We need the top PRECISION bits, including the "implicit"
				* one bit. signif will contain the top PRECISION bits.
				*/
				int shift = bitLen - Double.PRECISION;
				int shift2 = -shift;
				
				int highBits, lowBits;
				if ((shift & 0x1f) == 0) {
					highBits = mag[begin];
					lowBits = mag[begin + 1];
				} else {
					highBits = mag[begin] >>> shift;
					lowBits = (mag[begin] << shift2) | (mag[begin + 1] >>> shift);
					if (highBits == 0) {
						highBits = lowBits;
						lowBits = (mag[begin + 1] << shift2) | (mag[begin + 2] >>> shift);
					}
				}
				
				long signif = ((highBits & LONG_MASK) << 32) | (lowBits & LONG_MASK);
				signif ^= 1L << (Double.PRECISION - 1); // remove the implied bit
				
				long bits = ((long) Double.MAX_EXPONENT // add bias
				<< (Double.PRECISION - 1)) | signif;
				return new Floating(Double.longBitsToDouble(bits), bitLen - 1);
			}
			
			@Override
			public boolean equals(Object obj) {
				if (!(obj instanceof Floating f))
				return false;
				
				return this.signif == f.signif && this.exp == f.exp;
			}
			
			long toLong() {
				if (signif == 0.0)
				return 0L;
				
				long twoTo52 = 1L << (Double.PRECISION - 1);
				long bits = Double.doubleToRawLongBits(signif);
				bits = bits & (twoTo52 - 1) | twoTo52;
				
				long shift = exp - (Double.PRECISION - 1);
				return shift >= 0 ? bits << shift : bits >> -shift;
			}
			
			Floating nextUp() {
				if (signif == 0.0)
				return MIN_VALUE;
				
				double resSignif = Math.nextUp(signif);
				long resExp;
				
				if (resSignif >= 2.0) {
					resSignif /= 2.0;
					resExp = exp + 1;
				} else {
					resExp = exp;
				}
				
				return new Floating(resSignif, resExp);
			}
			
			Floating nextDown() {
				if (signif == 0.0 || this.equals(MIN_VALUE))
				return ZERO;
				
				double resSignif = Math.nextDown(signif);
				long resExp;
				
				if (resSignif < 1.0) {
					resSignif *= 2.0;
					resExp = exp - 1;
				} else {
					resExp = exp;
				}
				
				return new Floating(resSignif, resExp);
			}
			
			/**
			* @return this * 2^n
			*/
			Floating scaleByPowerOf2(long n) {
				if (signif == 0.0)
				return ZERO;
				
				return new Floating(signif, exp + n);
			}
			
			Floating add(Floating x) {
				if (this.signif == 0.0)
				return x;
				
				if (x.signif == 0.0)
				return this;
				
				Floating big, little;
				
				if (this.exp > x.exp) {
					big = this;
					little = x;
				} else {
					big = x;
					little = this;
				}
				
				// result == 2^big.exp * (big.signif + little.signif * 2^(little.exp - big.exp))
				long resExp = big.exp;
				double resSignif = big.signif;
				
				final long scale = little.exp - big.exp;
				// if (little.signif * 2^(little.exp - big.exp) < 2^-53)
				// then little is too small and does not affect the result
				if (scale >= -Double.PRECISION) {
					if (scale == 0) {
						resSignif += little.signif;
					} else {
						resSignif += little.signif * Double.parseDouble("0x1p" + scale);
					}
					
					// carry propagation
					if (resSignif >= 2.0) {
						resSignif /= 2.0;
						resExp++;
					}
				}
				
				return new Floating(resSignif, resExp);
			}
			
			Floating mul(Floating x) {
				if (this.signif == 0.0 || x.signif == 0.0)
				return ZERO;
				
				long resExp = this.exp + x.exp;
				double resSignif = this.signif * x.signif;
				
				// carry propagation
				if (resSignif >= 2.0) {
					resSignif /= 2.0;
					resExp++;
				}
				
				return new Floating(resSignif, resExp);
			}
			
			Floating div(Floating x) {
				if (this.signif == 0.0)
				return ZERO;
				
				long resExp = this.exp - x.exp;
				// Exponent underflow iff the exponents have different signs and
				// the sign of resExp is different from the sign of this.exp
				if (((this.exp ^ x.exp) & (this.exp ^ resExp)) < 0)
				return ZERO;
				
				double resSignif = this.signif / x.signif;
				
				// borrow propagation
				if (resSignif < 1.0) {
					if (resExp == Long.MIN_VALUE) // Exponent underflow
					return ZERO;
					
					resSignif *= 2.0;
					resExp--;
				}
				
				return new Floating(resSignif, resExp);
			}
			
			Floating sqrt() {
				if (this.signif == 0.0)
				return ZERO; 
				
				if ((exp & 1) == 1)
				return new Floating(Math.sqrt(2.0 * signif), (exp - 1) >> 1);
				
				return new Floating(Math.sqrt(signif), exp >> 1);
			}
		}
		
		public int bitLength() {
			if (bitLength == -1)
			bitLength = computeBitLength(mag, 0, mag.length);
			
			return bitLength;
		}
		
		private static int computeBitLength(int[] a, int offset, int len) {
			return len == 0 ? 0 : (len << 5) - Integer.numberOfLeadingZeros(a[offset]);
		}
		
		@Override
		public String toString() {
			return new BigInteger(toBinaryString(), 2).toString();
		}
		
		public String toBinaryString() {
			return toBinaryString(mag, 0, mag.length);
		}
		
		private static String toBinaryString(int[] mag, int begin, int end) {
			int len = end - begin;
			StringBuilder sb = new StringBuilder(len == 0 ? 1 : computeBitLength(mag, begin, len));
			
			if (len == 0) {
				sb.append('0');
			} else {
				sb.append(Integer.toBinaryString(mag[begin]));
				for (int i = begin + 1; i < end; i++) {
					sb.append("0".repeat(Integer.numberOfLeadingZeros(mag[i])));
					if (mag[i] != 0)
					sb.append(Integer.toBinaryString(mag[i]));
				}
			}
			
			return sb.toString();
		}
		
		private static int[] stripLeadingZeros(int[] a, int from) {
			for (; from < a.length && a[from] == 0; from++);
			return from == 0 ? a : from == a.length ? new int[0] : Arrays.copyOfRange(a, from, a.length);
		}
	}
\end{lstlisting}

\begin{lstlisting}[language = java]
	package sqrtTest;
	
	import java.math.BigInteger;
	import java.util.Random;
	import java.util.random.RandomGenerator;
	
	public class SqrtTest {
		
		public static void main(String[] args) {
			int nWords = 0b1000000000000000;
			Random rand = Random.from(RandomGenerator.getDefault());
			long t0, t1;
			long nTests = 1_000;
			long timeSum1 = 0, timeSum2 = 0;
			
			
			for (long n = 0; n < nTests; n++) {
				BigInteger x2 = new BigInteger(nWords << 5, rand);
				BigInt x1 = new BigInt(x2.toString());
				
				t0 = System.nanoTime();
				BigInt[] sqrtRem1 = x1.sqrtAndRemainder();
				t1 = System.nanoTime();
				long dt = t1 - t0;
				timeSum1 += dt;
				
				t0 = System.nanoTime();
				BigInteger[] sqrtRem2 = x2.sqrtAndRemainder();
				t1 = System.nanoTime();
				dt = t1 - t0;
				timeSum2 += dt;
				
				if (!sqrtRem1[0].toString().equals(sqrtRem2[0].toString()) || !sqrtRem1[1].toString().equals(sqrtRem2[1].toString())) {
					System.err.println("Inexact result:");
					System.err.println("sqrt = " + sqrtRem1[0] + ", rem = " + sqrtRem1[1]);
				}
			}
			
			double mean1 = (double) timeSum1 / nTests;
			double mean2 = (double) timeSum2 / nTests;
			System.out.println("avg time 1 = " + Math.round(mean1));
			System.out.println("avg time 2 = " + Math.round(mean2));
		}
	}
\end{lstlisting}

\subsection{Testing results}

Results of 1,000 tests for each input size:
\\

\begin{tabular}{|c||c|c|}
	\hline
	Input size (32-bit words) & paper's algorithm avg time (ns) & JDK22 avg time (ns) \\
	\hline
	\hline
	    0 &         583 &         3,128 \\
	    1 &      50,985 &        15,483 \\
	    2 &      53,034 &        22,855 \\
	    4 &      68,548 &        54,614 \\
	    8 &      86,643 &        93,587 \\
	   16 &      84,460 &       113,567 \\
	   32 &      89,293 &       136,668 \\
	   64 &      92,274 &       293,345 \\
	  128 &      96,131 &       957,569 \\
	  256 &     145,257 &     2,423,880 \\
	  512 &     240,063 &     4,986,413 \\
	 1024 &     676,682 &    13,758,715 \\
	 2048 &   2,014,007 &    34,831,680 \\
	 4096 &   7,150,738 &    99,797,550 \\
	 8192 &  25,017,604 &   284,323,520 \\
	16384 &  93,430,151 &   795,028,890 \\
	32768 & 399,697,826 & 2,297,631,960 \\
	\hline
\end{tabular}
\\

\begin{tikzpicture}
	\begin{axis} [
			title = The same data on cartesian chart,
			width = 12cm, height = 12cm,
			xlabel = Input size (32-bit words),
			ylabel = average running time ($10^{-9}$ s),
			no markers
		]
		\legend {
			\textsc{Paper's implementation},
			\textsc{JDK22 implementation}
		}
		\addplot[smooth, blue] coordinates {
			(0, 583)
			(1, 50985)
			(2, 53034)
			(4, 68548)
			(8, 86643)
			(16, 84460)
			(32, 89293)
			(64, 92274)
			(128, 96131)
			(256, 145257)
			(512, 240063)
			(1024, 676682)
			(2048, 2014007)
			(4096, 7150738)
			(8192, 25017604)
			(16384, 93430151)
			(32768, 399697826)
		};
		\addplot[smooth, red] coordinates {
			(0, 3128)
			(1, 15483)
			(2, 22855)
			(4, 54614)
			(8, 93587)
			(16, 113567)
			(32, 136668)
			(64, 293345)
			(128, 957569)
			(256, 2423880)
			(512, 4986413)
			(1024, 13758715)
			(2048, 34831680)
			(4096, 99797550)
			(8192, 284323520)
			(16384, 795028890)
			(32768, 2297631960)
		};
	\end{axis}
\end{tikzpicture}

\nocite{*}
\printbibliography[heading=bibintoc]

\end{document}